\documentclass[12pt]{article}
\usepackage{amsmath,amssymb,dsfont}
\usepackage{graphicx}

\newtheorem{theorem}{Theorem}

\newtheorem{definition}[theorem]{Definition}
\newtheorem{lemma}[theorem]{Lemma}

\newenvironment{proof}[1][Proof]{\noindent\textbf{#1.} }{\ }

\begin{document}
\title{Effect of phase shifts on EPR entanglement generated on two propagating Gaussian fields via coherent feedback}
\author{Zhan Shi and Hendra I. Nurdin 
\thanks{
Z. Shi and H. I. Nurdin are with School of Electrical Engineering and 
Telecommunications,  UNSW Australia,  
Sydney NSW 2052, Australia (e-mail: zhan.shi@student.unsw.edu.au,  h.nurdin@unsw.edu.au).} 
}
\maketitle

\begin{abstract}
Recent work has shown that deploying two nondegenerate optical parametric amplifiers (NOPAs) separately at two distant parties in a coherent feedback loop generates stronger Einstein-Podolski-Rosen (EPR) entanglement between two propagating continuous-mode output fields than a single NOPA under same pump power, decay rate and transmission losses. 
The purpose of this paper is to investigate the stability and EPR entanglement of a dual-NOPA coherent feedback system under the effect of phase shifts in the transmission channel between two distant parties. It is shown that, in the presence of phase shifts, EPR entanglement worsens or can vanish, but can be improved to some extent in certain scenarios by adding a phase shifter at each output with a certain value of phase shift. In ideal cases, in the absence of transmission and amplification losses, existence of EPR entanglement and whether the original EPR entanglement can be recovered by the additional phase shifters are decided by values of the phase shifts in the path.
\end{abstract}

\section{Introduction}
\label{sec:intro}
Entanglement is a key resource for quantum information processing. As an open quantum system is susceptible to external environment, entanglement would decay due to losses caused by unwanted interaction between the quantum system and its external electromagnetic field, which may lead to failure of quantum communication between two distant parties (Alice and Bob) and limit transmission distance \cite{GardinerBook}. Therefore, reliable generation and distribution of entanglement between two distant communicating parties (Alice and Bob) has become increasingly important.  Continuous-variable entanglement has an advantage over discrete-variable one due to its high efficiency in generation and measurement of quantum states \cite{BL2005}. As the most widely used continuous variable entangled resource, Gaussian EPR-like entangled pairs can be generated between amplitude and phase quadratures of two outgoing light beams of a nondegenerate optical parametric amplifier (NOPA) \cite{BSLR, Ou1992}. 

The main component of a NOPA is a two-end cavity which consists of a nonlinear $\chi^{(2)}$ crystal and mirrors. With a strong undepleted coherent pump beam employed as a source of energy, interactions between the pump beam and two modes inside the cavity generate a pair of outgoing beams in Gaussian EPR-like entangled states \cite{Ou1992}. In Fig \ref{fig:dual-nopa-cfb}, a NOPA is simply denoted by a block with four inputs and two outputs. More details of the NOPA are given in Section \ref{sec:system-model}.

\begin{figure}[htbp]
\begin{center}
\includegraphics[scale=0.5]{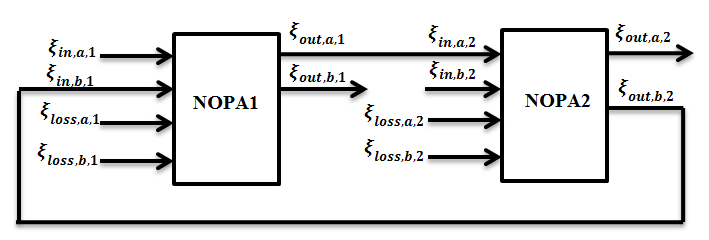}
\caption{A dual-NOPA coherent feedback system}\label{fig:dual-nopa-cfb}
\end{center}
\end{figure}

Our previous work \cite{SN2014} presents a dual NOPA coherent feedback system where two NOPAs are separately located at two distant endpoints (Alice and Bob) and connected in a feedback loop without employing any measurement devices, shown in Fig \ref{fig:dual-nopa-cfb}. 
In the network, two entangled outgoing fields $\xi_{out,a,2}$ and $\xi_{out,b,1}$ are generated. Our previous work \cite{SN2014} shows that under the same pump power, decay rate and transmission losses, the dual-NOPA coherent feedback network generates stronger EPR entanglement than a single NOPA placed in the middle of the two ends (at Charlie's). The paper also examines effects of losses and time delays on the dual-NOPA system. Not surprisingly,  EPR entanglement worsens as transmission and amplification losses increase; transmission time delays reduce the range of frequency over which EPR entanglement exists.

In this paper, we examine the effect of phase shifts along the transmission channels on EPR entanglement generated by the dual-NOPA coherent feedback system. What we are interested in is whether phase shifts degrade EPR entanglement; if they do, then whether we can recover it or minimize the EPR entanglement reduction by placing two adjustable phase shifters separately at each output. The paper is organised as follows. Section \ref{sec:prelim} briefly introduces linear quantum systems and an EPR entanglement criterion between two continuous-mode Gaussian fields. A description of our dual-NOPA coherent feedback system under influence of losses and phase shifts is given in Section \ref{sec:system-model}. Section \ref{sec:analysis} investigates the stability condition, as well as EPR entanglement under effects of phase shifts in a lossless system and a more general case where transmission losses and amplification losses are considered. Finally Section \ref{sec:conclusion} gives the conclusion of this paper.

\section{Preliminaries}
\label{sec:prelim}
This paper employs the following notations. $\imath$ denotes $\sqrt{-1}$, the transpose of a matrix of numbers or operators is denoted by $\cdot^T$ and $\cdot^*$ denotes (i) the complex conjugate of a number, (ii) the conjugate transpose of a matrix, as well as (iii) the adjoint of an operator. $I$ denotes an identity matrix. Trace operator is denoted by $\operatorname{Tr[\cdot]}$.

\subsection{Linear quantum systems}
\label{sec:linear-qsys}
An open linear quantum system without a scattering process contains $n$-bosonic modes $a_j(t)~(j=1,\ldots, n)$ satisfying $[a_i(t), a_j(t)^*]=\delta_{ij}$. The dynamics of the system can be described by the time-varying interaction Hamitonian between the system and environment
\begin{equation}
\label{interaction}
   H_{\rm int}(t) = \imath \sum_{j=1}^m (L_j\xi_j(t)^* - L_j^* \xi_j(t)), 
\end{equation}
in which $L_j$ is the $j$-th system coupling operator and  $\xi_j(t)~(j=1,\ldots,m)$ is the field operator describing the $j$-th environment field  \cite{GardinerBook}. When the environment is under the condition of the Markov limit, the field operator under the vacuum state satisfies $[\xi_j(t), \xi_j(s)^*]=\delta(t-s)$, where $\delta(t)$ denotes the Dirac delta function. 
When $L_j$ is linear and $H$ is quadratic in $a_j$ and $a_j^*$, the Heisenberg evolutions of mode $a_j$ and output filed operator $\xi_{out,j}$ are defined by $a_j(t)=U(t)^* a_j U(t)$  and $\xi_{out,j}(t)=U(t)^*\xi_{in,j}(t)U(t)$ with unitary  $U(t)={\rm exp}^{\hspace{-0.5cm}\longrightarrow}~
(-i\int_0^t H_{\rm int}(s)ds)$ and is of the form
\begin{eqnarray}
    \dot{z}(t)&=&Az(t)+B\xi(t), \label{eq:dynamics} \\
     \xi_{out,j}(t)&=&Cz(t)+D\xi(t), \label{eq:output}
\end{eqnarray}
for some real matrices $A$, $B$, $C$ and $D$, where we have defined 
\begin{eqnarray}
    z&=&(a_1^q, a_1^p, \ldots, a_n^q, a_n^p)^T, \nonumber\\
    \xi &=&(\xi_{1}^q, \xi_{1}^p, \ldots, \xi_{m}^q, \xi_{m}^p)^T, \nonumber\\
    \xi_{out}&=&(\xi_{out,1}^q, \xi_{out,1}^p, \ldots, \xi_{out,l}^q, \xi_{out,l}^p)^T,
\end{eqnarray}
with {\it quadratures} \cite{Belavkin2008,WisemanBook}
\begin{eqnarray}
a_j^q &=& a_j+a_j^*, \quad a_j^p = (a_j-a_j^*)/i, \nonumber \\
\xi_j^q &=& \xi_j+\xi_j^*, \quad \xi_j^p = (\xi_j-\xi_j^*)/i. \label{eq:quadratures}
\end{eqnarray}

\vspace{6pt}
\subsection{EPR entanglement between two continuous-mode fields}
\label{sec:entanglement}
Unlike bipartite entanglement of two-mode Gaussian states, which can be measured by the logarithmic negativity \cite{Laurat2005}, EPR entanglement between two continuous-mode (many-mode) output fields $\xi_{out,1}$ and $\xi_{out,2}$ has to be evaluated in frequency domain \cite{Ou1992,BL2005,Vitali2006}. $F(\imath\omega)$, the Fourier transform of $f(t)$, can be achieved by Fourier transformation $F\left(\imath\omega\right)=\frac{1}{\sqrt{2\pi}}\int_{-\infty}^{\infty} f\left(t\right)e^{-\imath\omega t} dt$. Similarly, we can get the Fourier transforms of 
$\xi_{out,1}(t)$, $\xi_{out,2}(t)$, $z(t)$ and $\xi(t)$, 
as  $\tilde \Xi_{out,1}\left(\imath\omega\right)$, $\tilde \Xi_{out,2}\left(\imath\omega\right)$, $Z(\imath \omega)$  and  $\Xi(\imath \omega)$, respectively. 

Based on equations (\ref{eq:dynamics}) and (\ref{eq:output}), we have
\begin{align}
\tilde \Xi_{out,1}^q(\imath \omega)+\tilde \Xi_{out,2}^q(\imath \omega)=C_1 Z\left(\imath\omega\right)+ D_1\Xi\left(\imath\omega\right), \nonumber\\
\tilde \Xi_{out,1}^p(\imath \omega)-\tilde \Xi_{out,2}^p(\imath \omega)=C_2 Z\left(\imath\omega\right)+ D_2\Xi\left(\imath\omega\right),
\end{align}
where $ C_1=[1\ 0\ 1\ 0]C$, $C_2=[0\ 1\ 0\ {-}1]C$, $D_1=[1\ 0\ 1\ 0]D$ and $D_2=[0\ 1\ 0\ {-}1]D$.

If the ingoing signals are in a vacuum state, the EPR entanglement between the two fields are related to the two-mode amplitude squeezing spectra $V_{+}$ and the two-mode phase squeezing spectra $V_{-}$ which have the following definitions
\begin{align}
& \quad \langle (\tilde \Xi_{out,1}^q(\imath \omega)+\tilde \Xi_{out,2}^q(\imath \omega))^* (\tilde \Xi_{out,1}^q(\imath \omega')+\tilde \Xi_{out,2}^q(\imath \omega')) \rangle = V_+(\imath \omega)\delta(\omega-\omega'), \nonumber \\
& \quad  \langle (\tilde \Xi_{out,1}^p(\imath \omega)-\tilde \Xi_{out,2}^p(\imath \omega))^* (\tilde \Xi_{out,1}^p(\imath \omega')-\tilde \Xi_{out,2}^p(\imath \omega')) \rangle = V_-(\imath \omega) \delta(\omega-\omega'),
\end{align}
where $\langle \cdot \rangle$ denotes quantum expectation. 
$V_+(\imath \omega)$ and $V_-(\imath \omega)$ are real valued and can be easily calculated by using the transfer functions $H_j(\imath\omega)=C_{j}\left(\imath\omega I-A \right)^{-1}B+D_{j}$  ($j=1,2$), as described in \cite{NY2012,GJN2010},
\begin{align}
V_+(\imath\omega)=& {\rm Tr}\left[H_1(\imath\omega)^* H_1(\imath\omega)\right], \label{eq:V_+}\\
V_-(\imath\omega)=& {\rm Tr}\left[H_2(\imath\omega)^* H_2(\imath\omega)\right]. \label{eq:V_-}
\end{align}
Denote $V(\imath\omega) = V_+(\imath\omega)+V_-(\imath\omega)$. The sufficient condition that the fields $\xi_{out,1}$ and  $\xi_{out,2}$ are correlated at the frequency $\omega$ rad/s is \cite{Vitali2006},
\begin{align}
V(\imath\omega)< 4. \label{eq:entanglement-criterion}
\end{align}
Ideally, we would like $V(\imath\omega) = V_+(\imath\omega) = V_-(\imath\omega)= 0$ for all $\omega$, which denotes infinite-bandwidth two-mode squeezing, representing an ideal Einstein-Podolski-Rosen state. However, in reality the ideal EPR correlation can not be achieved,  so in practice the goal is to make $V(\imath \omega)$ as small as possible over a wide frequency range \cite{Vitali2006}. Following \cite{SN2014} and \cite{NY2012}, for low frequencies, we have a good approximation that $V_{+}(i\omega) \approx V_{+}(0)$ and $V_{-}(i\omega) \approx V_{-}(0)$.

Define $\xi^{\psi_1}_{out,1}=e^{\imath \psi_1}\xi_{out,1}$, $\xi^{\psi_2}_{out,2}=e^{\imath \psi_2}\xi_{out,2}$ with $\psi_1, \psi_2 \in(-\pi,\pi]$ and denote the corresponding two-mode squeezing spectra as $V^{\psi_1, \psi_2}_\pm(\imath\omega,\psi_1, \psi_2)$, we have the following definition of EPR entanglement.
\begin{definition}
Fields $\xi_{out,1}$ and  $\xi_{out,2}$ are EPR entangled at the frequency $\omega$ rad/s if $\exists ~\psi_1, \psi_2 \in(-\pi,\pi]$ such that
\begin{align}
V^{\psi_1, \psi_2}(\imath\omega,\psi_1, \psi_2) = V^{\psi_1, \psi_2}_+(\imath\omega,\psi_1, \psi_2)+V^{\psi_1, \psi_2}_-(\imath\omega,\psi_1, \psi_2)< 4. \label{eq:entanglement-criterion-2}
\end{align}
Unless otherwise specified, throughout the paper EPR entanglement refers to the case with $\psi_1=\psi_2=0$.
EPR entanglement is said to vanish at frequency $\omega$ if there are no values of $\psi_1$ and $\psi_2$ satisfying the above criterion.
\end{definition}
  
\section{The system model}
\label{sec:system-model}
In this section, we consider a dual-NOPA ($G_1$ and $G_2$) coherent feedback network shown in Fig. \ref{fig:dual-nopa-cfb-ps} previously proposed in \cite{SN2014}. During the transmission, the system undergoes transmission losses and possibly some phase shifts. Transmission loss in each path of the system is modelled by a beamsplitter with transmission rate $\alpha$ and reflection rate $\beta$ ($\alpha^2+\beta^2=1$), and phase shift  $\theta_i$ $(i=1,2)$ in each path is modelled by a phase shifter, whose outgoing field $\xi_{out}$ and input field $\xi_{in}$ have the relation $ \xi_{out}=e^{\imath \theta_i} \xi_{in}$  \cite{NJD2009, GK2005}.  Two phase shifters with adjustable phase shifts $\phi_1$ and $\phi_2$ are placed at two outputs separately. We are interested in EPR entanglement generated between continuous-mode outgoing fields $\xi_{out,a,2}$ and $\xi_{out,b,1}$.

\begin{figure*}[htbp]
\begin{center}
\includegraphics[scale=0.5]{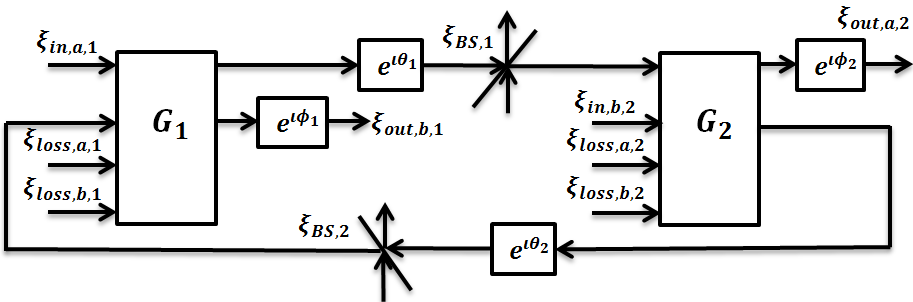}
\caption{A dual-NOPA coherent feedback system under effects of losses and phase shifts.}\label{fig:dual-nopa-cfb-ps}
\end{center}
\end{figure*}
Each NOPA $G_i$ $(i=1,2)$ has two oscillator modes $a_i$ and $b_i$ inside its cavity. As a strong coherent beam is pumped to the nonlinear $\chi^{(2)}$ crystal in the cavity, the modes $a_i$ and $b_i$ are coupled via a two-mode squeezing Hamiltonian $H= \frac{\imath}{2} \epsilon\left( a_i^* b_i^*- a_ib_i\right)$, where $\epsilon$ is a real coupling coefficient related to the amplitude of the pump beam \cite{Ou1992}. Mode $a_i$ is coupled to ingoing noise $\xi_{in,a,i}$ and amplification loss $\xi_{loss,a,i}$ via coupling operators  $L_1=\sqrt{\gamma}a_i$ and $L_3=\sqrt{\kappa}a_i$, respectively, for some constant damping rates $\gamma$ and $\kappa$; similarly mode $b_i$  interacts with input signal $\xi_{in,b,i}$ and additional noise $\xi_{loss,b,i}$ by operators $L_2=\sqrt{\gamma}b_i$ and $L_4=\sqrt{\kappa}b_i$. The modes satisfy the commutation relations $[a_i, a_j^*]=\delta_{ij}$, $[a_i, b_j]=0$,  $[a_i, b_j^*]=0$,  $[a_i, a_j]=0$ and $[b_i, b_j]=0$ $(i,j=1,2)$ \cite{GardinerBook}.

The dynamics of the system in Fig. \ref{fig:dual-nopa-cfb-ps} is given by
\small
\begin{align}
\dot{a_1}\left(t\right)=&-\frac{\gamma+\kappa}{2}a_1\left(t\right)+\frac{\epsilon} {2}b^*_1\left(t\right)-\sqrt{\gamma}\xi_{in,a,1}\left(t\right)-\sqrt{\kappa}\xi_{loss,a,1}\left(t\right), \nonumber \\ 
\dot{b_1}\left(t\right)=&-\frac{\gamma+\kappa}{2}b_1\left(t\right)+\frac{\epsilon}{2}a^*_1\left(t\right)-\alpha \gamma e^{\imath \theta_2}b_2\left(t\right) -\alpha \sqrt{\gamma}e^{\imath \theta_2} \xi_{in,b,2}\left(t\right) -\sqrt{\kappa}\xi_{loss,b,1}\left(t\right) -\beta \sqrt{\gamma}\xi_{BS,2}\left(t\right), \nonumber \\ 
\dot{a_2}\left(t\right)=&-\frac{\gamma+\kappa}{2}a_2\left(t\right)+\frac{\epsilon}{2}b^*_2\left(t\right)-\alpha \gamma e^{\imath \theta_1}a_1\left(t\right) -\alpha \sqrt{\gamma}e^{\imath \theta_1} \xi_{in,a,1}\left(t\right) -\sqrt{\kappa}\xi_{loss,a,2}\left(t\right) -\beta \sqrt{\gamma}\xi_{BS,1}\left(t\right), \nonumber \\ 
\dot{b_2}\left(t\right)=&-\frac{\gamma+\kappa}{2}b_2\left(t\right)+\frac{\epsilon}{2}a^*_2\left(t\right)-\sqrt{\gamma}\xi_{in,b,2}\left(t\right) -\sqrt{\kappa}\xi_{loss,b,2}\left(t\right), \label{eq:NOPAS-internal}
\end{align}
\normalsize
with outputs
\small
\begin{align}
\xi_{out,b,1}\left(t\right)=&\sqrt{\gamma}e^{\imath \phi_1} b_1\left(t\right)+\alpha\sqrt{\gamma}e^{\imath (\theta_2+\phi_1)} b_2\left(t\right)+ \alpha e^{\imath (\theta_2+\phi_1)} \xi_{in,b,2}\left(t\right) + \beta e^{\imath \phi_1}\xi_{BS,2}\left(t\right) ,\nonumber \\
\xi_{out,a,2}\left(t\right)=&\sqrt{\gamma}e^{\imath \phi_2} a_2\left(t\right)+\alpha\sqrt{\gamma}e^{\imath (\theta_1+\phi_2)} a_1\left(t\right)+ \alpha e^{\imath (\theta_1+\phi_2)} \xi_{in,a,1}\left(t\right) + \beta e^{\imath \phi_2}\xi_{BS,1}\left(t\right). \label{eq:NOPAs-out}
\end{align}
\normalsize
Define the quadratures
\begin{align}
z=&[a^q_1, a^p_1, b^q_1, b^p_1,a^q_2, a^p_2, b^q_2, b^p_2]^T,\nonumber \\
\xi=&[\xi^q_{in,a,1},\xi^p_{in,a,1},\xi^q_{in,b,2},\xi^p_{in,b,2},\xi^q_{loss,a,1},\xi^p_{loss,a,1},\xi^q_{loss,b,1},\xi^p_{loss,b,1},\nonumber\\
&\xi^q_{loss,a,2},\xi^p_{loss,a,2},\xi^q_{loss,b,2},\xi^p_{loss,b,2},\xi^q_{BS,1},\xi^p_{BS,1},\xi^q_{BS,2},\xi^p_{BS,2}]^T,\nonumber\\
\xi_{out}=&[\xi^q_{out,b,1},\xi^p_{out,b,1},\xi^q_{out,a,2},\xi^p_{out,a,2}]^T.
\end{align}
According to (\ref{eq:dynamics}), (\ref{eq:output})， (\ref{eq:NOPAS-internal}) and (\ref{eq:NOPAs-out}), we have
\begin{eqnarray}
\dot{z}\left(t\right) &=& A z\left(t\right)+ B \xi\left(t\right), \label{eq:dual-NOPA-dynamics-1} \\
\xi_{out}\left(t\right) &=& C z\left(t\right)+D\xi\left(t\right),\label{eq:dual-NOPA-dynamics-2}
\end{eqnarray}
where $A, B, C$, and $D$ are real matrices. 

As mentioned in Section \ref{sec:entanglement}, two-mode squeezing spectra $V_{\pm}(\imath\omega)$ can be approximated to $V_{\pm}(0)$ at $\omega=0$ in low frequency domain. In the remainder of this paper, we evaluate degree of EPR entanglement between outgoing fields $\xi_{out,a,2}$ and $\xi_{out,b,1}$ by $V_{\pm}(0)$. Based on (\ref{eq:V_+}) and (\ref{eq:V_-}), we get
\begin{eqnarray}
V(0) = V_+(0)+V_-(0)={\rm Tr}\left[H_1^{*}H_1+H_2^{*}H_2\right],\label{eq:entanglement}
\end{eqnarray}
where $H = D-CA^{-1}B$, $H_1 = \left[\begin{array}{cccc} 1 & 0 & 1 & 0 \end{array}\right]H$ and  $ H_2 = \left[\begin{array}{cccc} 0 & 1 & 0 & -1 \end{array}\right]H$.

\section{Analysis of effect of phase shifts on the dual-NOPA coherent feedback system}
\label{sec:analysis}
Here we analyse effects of phase shifts $\theta_1$ and $\theta_2$ on stability and EPR entanglement of the dual-NOPA coherent feedback system. We investigate EPR entanglement when the system is lossless, that is, amplification and transmission losses are neglected, as well as EPR entanglement of the system with losses. Moreover, we examine effects of adjustable phase shifters with phase shifts $\phi_1$ and $\phi_2$ to see whether they can recover the EPR entanglement impacted by $\theta_1$ and $\theta_2$.  

Parameters of the system are defined as follows. Based on \cite{SN2014} and \cite{Iida2012}, we define $\gamma_r=7.2 \times 10^7$ Hz as a reference value of the transmissivity mirrors, $\epsilon=x\gamma_r$ Hz and $\gamma= \frac{\gamma_r}{y}$ Hz, where $x$ and $y$ ($0<x,y \leq 1$) are adjustable real parameters. Following \cite{SN2014, Iida2012}, we assume that $\kappa= \frac{3 \times 10^6}{\sqrt{2}}$ when $\epsilon=0.6 \gamma_r$  and the value of $\kappa$ is proportional to the absolute value of $\epsilon$, so we set $\kappa=\frac{3 \times 10^6}{\sqrt{2} \times 0.6}x$. Transmission rate $\alpha \in (0,1]$ and reflection rate $\beta=\sqrt{1-\alpha^2}$. Range of phase shifts $\theta_1, \theta_2, \phi_1, \phi_2$ is $(-\pi, \pi]$. Note that we employ Mathematica to perform the complex symbolic manipulations that are required in this paper. 

\subsection{Stability condition}
\label{sec:stability}
To make the system workable, stability must be guaranteed. In our case, the system is stable which means that as time goes to infinity, the mean total number of photons within cavities of the two NOPAs must not increase continuously. Mathematically, stability condition holds when matrix $A$ in equation (\ref{eq:dual-NOPA-dynamics-1}) is Hurwitz, that is, real parts of all eigenvalues of $A$ are negative. Based on this, we have the following theorem which states the stability condition of our system with parameters $x$, $y$, $\kappa$, $\alpha$, $\theta_1$ and $\theta_2$.

\begin{theorem}
\label{th: stability} The dual-NOPA coherent feedback system under the influence of losses and phase shifts is stable if and only if 
\small
\begin{equation}
 x y< \frac{\left(1+\frac{y \kappa}{\gamma_r}\right)^2}{\sqrt{\left(1+\frac{y \kappa}{\gamma_r}\right)^2 + \alpha^2 }+\alpha  \left\vert  \cos \frac{\Delta \theta}{2} \right\vert}  \label{eq:stability-cond}
\end{equation}
with  $0<x, y\leq1$, $0\leq \alpha \leq1$ and $\Delta \theta=\theta_1-\theta_2$.
\end{theorem}
\normalsize

\begin{proof}
Based on (\ref{eq:NOPAS-internal}) and (\ref{eq:dual-NOPA-dynamics-1}), we have 

\scalebox{0.9}{%
 \begin{minipage}{1.0\linewidth}
 \begin{displaymath}
A =
\left[\begin{array}{cccccccc}
-\frac{\gamma+\kappa}{2} & 0 & \frac{\epsilon}{2} & 0 & 0 & 0 & 0 & 0 \\
0 & -\frac{\gamma+\kappa}{2} & 0 & -\frac{\epsilon}{2} & 0 & 0 & 0 & 0 \\
\frac{\epsilon}{2} & 0 & -\frac{\gamma+\kappa}{2} & 0 & 0 & 0 & -\alpha \gamma \cos \theta_2 & \alpha \gamma \sin \theta_2 \\
0 & -\frac{\epsilon}{2} & 0 & -\frac{\gamma+\kappa}{2} & 0 & 0 & -\alpha \gamma \sin \theta_2 & -\alpha \gamma \cos \theta_2 \\
-\alpha \gamma \cos \theta_1 & \alpha \gamma \sin \theta_1 & 0 & 0 & -\frac{\gamma+\kappa}{2} & 0 & \frac{\epsilon}{2} & 0 \\
-\alpha \gamma \sin \theta_1 & -\alpha \gamma \cos \theta_1 & 0 & 0 & 0 & -\frac{\gamma+\kappa}{2} & 0 & -\frac{\epsilon}{2} \\
0 & 0 & 0 & 0 & \frac{\epsilon}{2} & 0 & -\frac{\gamma+\kappa}{2} & 0  \\
0 & 0 & 0 & 0 & 0 & -\frac{\epsilon}{2} & 0 & -\frac{\gamma+\kappa}{2} 
\end{array} \right]. 
\end{displaymath}
 \end{minipage}
}\\
\vspace{6pt}
Eigenvalues of the matrix are 
\small
\begin{eqnarray}
\lambda&=& -\frac{\gamma+\kappa}{2} \pm \sqrt{\frac{\epsilon^2}{4} \pm \frac{\alpha \epsilon \gamma}{2} \cos \frac{\Delta \theta}{2} \pm \imath \frac{\alpha \epsilon \gamma}{2} \sin \frac{\Delta \theta}{2}}  \nonumber \\
&=& -\frac{\gamma+\kappa}{2} \pm  \sqrt{r e^{\imath \varphi }},
\end{eqnarray}
\normalsize
where
\small 
\begin{eqnarray}
r &=& \sqrt{\left( \frac{\epsilon^2}{4} \pm \frac{\alpha \epsilon \gamma}{2} \cos \frac{\Delta \theta}{2} \right)^2 + \left( \frac{\alpha \epsilon \gamma}{2} \sin \frac{\Delta \theta}{2} \right)^2},  \nonumber \\
\varphi &=& \arctan {\frac{ \frac{\alpha \epsilon \gamma}{2} \sin \frac{\Delta \theta}{2} }{\frac{\epsilon^2}{4} \pm \frac{\alpha \epsilon \gamma}{2} \cos \frac{\Delta \theta}{2} }} + 2k \pi, k \in \mathbb{Z}. 
\end{eqnarray}
\normalsize
Real parts of the eigenvalues are 
\small
\begin{eqnarray}
\operatorname{real}(\lambda) &=& -\frac{\gamma+\kappa}{2} \pm  \sqrt{r} \cos \frac{\varphi}{2} \nonumber \\
&=& -\frac{\gamma+\kappa}{2} \pm  \sqrt{r} \sqrt{\frac{1+ \cos \varphi}{2}} \nonumber \\
&=& -\frac{\gamma+\kappa}{2} \pm \frac{1}{2} \left(\frac{\epsilon^2}{2} \pm \alpha \epsilon \gamma \cos \frac{\Delta \theta}{2}  + \sqrt{\frac{\epsilon^4}{4} + \alpha^2 \epsilon^2 \gamma^2  \pm \alpha \epsilon^3 \gamma \cos \frac{\Delta \theta}{2}} \right)^{\frac{1}{2}}. 
\end{eqnarray}
\normalsize
Hence,
\small
\begin{eqnarray}
&&\max (\operatorname{real}(\lambda))=  -\frac{\gamma+\kappa}{2} + \frac{1}{2} \left(\frac{\epsilon^2}{2} + \alpha \epsilon \gamma \left\vert \cos \frac{\Delta \theta}{2} \right \vert + \sqrt{\frac{\epsilon^4}{4} + \alpha^2 \epsilon^2 \gamma^2  + \alpha \epsilon^3 \gamma \left\vert \cos \frac{\Delta \theta}{2}\right\vert }  \right)^{\frac{1}{2}}
\end{eqnarray}
\normalsize
Stability holds when $\max (\operatorname{real}(\lambda))<0$. By solving the inequality with $\epsilon = x \gamma_r$, $\gamma = \frac{\gamma_r}{y}$, $0<x, y, \alpha \leq 1$, the theorem is obtained. 
\end{proof}

The theorem directly shows that, stability of the system is only impacted by the difference between values of $\theta_1$ and $\theta_2$, not by the values of $\theta_1$ and $\theta_2$ individually. However, as $0 \leq  \left\vert  \cos \frac{\Delta \theta}{2} \right\vert \leq 1$ and $\alpha$ has positive value, as long as the system without phase shifts is stable, the system maintains stability in the presence of phase shifts due to the transmission distance.

\subsection{Effect of phase shifts on EPR entanglement of a lossless system}
\label{sec:entanglement-lossless}
In this part, we investigate the effect of the phase shifts on the EPR entanglement between $\xi_{out,a,2}$ and $\xi_{out,b,1}$ when the system has no transmission losses ($\alpha=1$) and no amplification losses ($\kappa=0$). Based on Section \ref{sec:system-model}, we obtain the two-mode squeezing spectra between the two outgoing fields of the dual-NOPA coherent feedback system $V_\pm(0, \theta_1, \theta_2, \phi)$ as a function of $\theta_1$, $\theta_2$ and $\phi=\phi_1+\phi_2$ at $\omega=0$ when $\alpha=1$ and $\kappa=0$,
\small
\begin{eqnarray}
&& V_\pm(0, \theta_1, \theta_2, \phi)=2\frac{b_1+2b_2\cos \left(\frac{\theta_1-\theta_2}{2} \right)\cos \left( \frac{\theta_1+\theta_2}{2}+\phi\right)+b_3 \cos\left(\theta_1-\theta_2 \right)}{b_4-b_3\cos \left(\theta_1-\theta_2\right)}, \nonumber \\  \label{eq:V_lossless}
\end{eqnarray}
\normalsize
where
\begin{eqnarray}
b_1 &=& \epsilon ^8 + 12 \epsilon ^6 \gamma ^2 - 10\epsilon ^4 \gamma ^4 + 12\epsilon ^2 \gamma ^6 + \gamma ^8 , \nonumber \\
b_2 &=& 4\epsilon \gamma (\epsilon ^2 - \gamma ^2)(\epsilon ^2 + \gamma ^2)^2,\nonumber \\
b_3 &=& 8\epsilon ^2 \gamma ^2 (\epsilon ^2 - \gamma ^2)^2,\nonumber \\
b_4 &=& \epsilon ^8 - 4\epsilon ^6 \gamma ^2 + 22\epsilon ^4 \gamma ^4 - 4\epsilon ^2 \gamma ^6 + \gamma ^8.\label{eq:b1-5}
\end{eqnarray}

What is of our interest is whether $\theta_1$ and $\theta_2$ decrease the degree of EPR entanglement; if they do, whether $\phi$ can recover the original EPR entanglement (when $\theta_1=\theta_2=0$) or at least improve the EPR entanglement, as well as how much EPR entanglement can be improved. To this end, we define the following functions at $\omega=0$,
\begin{itemize}
\item $V_\pm^{nops}$, the two-mode squeezing spectra between the two outgoing fields of the dual-NOPA coherent feedback system without phase shifts. That is, $V_\pm(0, \theta_1, \theta_2, \phi)$ is as in (\ref{eq:V_lossless}) at $\theta_1=\theta_2=\phi_1=\phi_2=0$;
\item $V_\pm^{ps}(\theta_1,\theta_2)$, the two-mode squeezing spectra between the two outgoing fields of the dual-NOPA coherent feedback system under the effect of the phase shifts $\theta_1$ and $\theta_2$, but without $\phi_1$ and $\phi_2$. That is, $V_\pm(0, \theta_1, \theta_2, \phi)$ is as in (\ref{eq:V_lossless}) at $\phi=0$;
\item $V_\pm(\phi)$, the two-mode squeezing spectra between the two outgoing fields of the dual-NOPA coherent feedback system under  the effect of phase shifts $\phi_1$ and $\phi_2$ with fixed values of $\theta_1$ and $\theta_2$. That is, $V_\pm(0, \theta_1, \theta_2, \phi)$ is as in (\ref{eq:V_lossless}) for fixed $\theta_1$ and $\theta_2$;
\item $f(\theta_1,\theta_2)=V_\pm^{ps}(\theta_1,\theta_2)- V_\pm^{nops}$. If $\theta_1$ and $\theta_2$ degrade the EPR entanglement, then $f(\theta_1,\theta_2)>0$;
\item $g(\phi)=V_\pm(\phi)- V_\pm^{nops}$. If the EPR entanglement degraded by a fixed value of $\theta_1$ and $\theta_2$ is fully recovered by $\phi_1$ and $\phi_2$, then $g(\phi)=0$;
\item $h(\phi)=V_\pm(\phi)- V_\pm^{ps}$. If the EPR entanglement impacted by $\theta_1$ and $\theta_2$ is improved by $\phi_1$ and $\phi_2$, then $h(\phi)<0$.
\end{itemize}
\vspace{6pt}
\subsubsection{A simple case ($\theta_1=\theta_2=\theta$)}
\label{sec:simple_case}
Let us begin with a simple case, where phase shifts $\theta_1=\theta_2=\theta$.
According to (\ref{eq:V_lossless}), we have
\begin{eqnarray}
f(\theta,\theta) &=& \frac{4b_2(\cos (\theta)-1)}{b_4-b_3}.  \label{eq:f_simple_case}
\end{eqnarray}
Based on the stability condition (\ref{eq:stability-cond}), here system is stable when $xy<\sqrt{2}-1$, that is, $\epsilon< (\sqrt{2}-1))\gamma$, hence $b_2<0$. Moreover $b_4-b_3=(\epsilon^2-2\epsilon \gamma- \gamma^2)^2(\epsilon^2+ 2\epsilon\gamma-\gamma^2)^2 >0$. Therefore $f(\theta, \theta) \geq 0$ (equality holds when $\theta=0$), which implies EPR entanglement worsens in the presence of phase shifts $\theta_1$ and $\theta_2$.
Now we examine the effect of $\phi$. We have
\begin{eqnarray}
g(\phi) &=& \frac{4b_2(\cos (\theta+\phi)-1)}{b_4-b_3},  \label{eq:g_simple_case}
\end{eqnarray}
we can see that as long as $\phi= \lbrace -\theta, \pm\pi -\theta \rbrace$, the EPR entanglement is fully recovered.
\vspace{6pt}


\subsubsection{General case}
\label{sec:general_case_lossless}
Here we consider the lossless system in a general situation, where phase shifts $\theta_1$ and $\theta_2$ can be different. 

Let $\kappa=0$, $\alpha=1$, $m=\frac{\theta_1-\theta_2}{2}$, $n=\frac{\theta_1+\theta_2}{2}$, $m, n, m+n, n-m \in (-\pi, \pi]$ and $\phi=\phi_1+\phi_2$, $\phi\in (-2\pi, 2\pi]$. Then (\ref{eq:V_lossless}) can be written as
\begin{eqnarray}
&&V_\pm(0, m, n, \phi) =  2\frac{b_1+2b_2\cos \left(m \right)\cos \left( n +\phi\right)+b_3 \cos\left(2m \right)}{b_4-b_3\cos \left(2m\right)} \label{eq:V_lossless_mn}
\end{eqnarray}
where $b_i$ ($i=1,2,3,4,5$) is as in (\ref{eq:b1-5}).

Analysing (\ref{eq:V_lossless_mn}) gives the lemmas below.
\begin{lemma}
\label{le: f_lossless} 
The presence of the phase shifts $\theta_1\neq0$ and $\theta_2\neq0$ degrades the two-mode squeezing spectra ($V_\pm^{ps}(m,n)>V_\pm^{nops}$), thus degree of EPR entanglement becomes worse or EPR entanglement may vanish.
\end{lemma}
\begin{proof}
Based on the functions defined at beginning of this subsection, we have 
\small
\begin{eqnarray}
f(m,n) &=& 2\frac{b_1+ 2b_2 \cos (m) \cos(n)+ b_3 \cos(2m)}{b_4-b_3\cos (2m)}-2\frac{b_1+ 2b_2+ b_3 }{b_4-b_3}, \label{eq:f_lossless} \\
\frac{\partial f}{\partial n}(m,n) &=& \frac{-4 b_2 \cos(m) \sin(n)}{b_4 - b_3 \cos(2 m)}, \label{eq:f1_lossless} \\
\frac{\partial f}{\partial m}(m,n) &=& -4 \sin(m) \frac{ b_2 \cos(n) (b_4 + b_3 \sin(m)^2)}{(b_4 - b_3 \cos(2 m))^2}-4 \sin(m) \frac{ 3 b_2 b_3 \cos(m)^2 \cos(n) }{(b_4 - b_3 \cos(2 m))^2} \nonumber \\
&& \quad -4 \sin(m) \frac{2 b_3 (b_1 + b_4) \cos(m)}{(b_4 - b_3 \cos(2 m))^2}. \label{eq:f2_lossless}  
\end{eqnarray}
\normalsize
$f(m,n)$ is a periodic continuous twice differentiable function with variables $m$ and $n$, and it is convenient to take the range of $m$ and $n$ to be the entire real line. Hence global minima of $f$ must be stationary points. Therefore for $m, n, m+n, n-m \in (-\pi, \pi]$, global minima of $f(m,n)$ are stationary points as well. With the help of Mathematica, we obtain that the first order partial derivatives of $f$ with respect to variables $m$ and $n$ vanish at $(m,n)=\lbrace(0, 0), (0, \pi), (\pm \frac{\pi}{2}, \frac{\pi}{2}), (- \frac{\pi}{2}, -\frac{\pi}{2}) \rbrace$. The values of $f$ at these stationary points are
\small
\begin{eqnarray}
&&f(m,n)= \left\{ \begin{array}{l}
0, \quad \textrm{if $(m,n)=(0, 0)$, that is, $\theta_1=\theta_2=0$} \\
\frac{-8b_2}{b_4-b_3}, \quad \textrm{if $(m,n)=(0, \pi)$} \\
\frac{-4(b_1b_3+b_3b_4+b_2b_3+b_2b_4)}{(b_4-b_3)(b_4+b_3)}, \quad \textrm{if $(m,n)=\lbrace(\pm \frac{\pi}{2}, \frac{\pi}{2}), (- \frac{\pi}{2}, -\frac{\pi}{2}) \rbrace$}.
\end{array} \right.
\end{eqnarray}
\normalsize
Again, when the system is stable, we have as before that $b_2<0$, $b_3>0$ and $b_4-b_3 >0$. Moreover, Mathematica shows that $b_1b_3+b_3b_4+b_2b_3+b_2b_4=4\epsilon\gamma(\epsilon^2-\gamma^2)(\epsilon^2+2\epsilon\gamma-\gamma^2)^2(\epsilon^2+\gamma^2)^4<0$. Thus, $f(m,n)=V_\pm^{ps}(m,n)- V_\pm^{nops} \geq 0$. Equality holds when $(m,n)=(0, 0)$, which is the case with no phase shifts. We obtain Lemma \ref{le: f_lossless}.
\end{proof}
\begin{lemma}
\label{le: phi_lossless} 
When $m \neq \pm \frac{\pi}{2}$, $\phi$ ensures the existence of EPR entanglement and minimizes the EPR  entanglement reduction caused by $\theta_1$ and $\theta_2$ if its value is set as $\phi_0$ \footnote[1]{Eventhough there is more than one minimum $\phi_0$, the function $V_\pm^{im}(m)$ takes the same value for all the minima.},
\begin{eqnarray}
\phi_0 = \left\{ \begin{array}{ll}
-n & \textrm{if $m \in (-\frac{\pi}{2}, \frac{\pi}{2})$}\\
\pm \pi - n & \textrm{if $m \in (-\pi, -\frac{\pi}{2})  \cup (\frac{\pi}{2}, \pi] $}. 
\end{array} \right.\label{eq: phi0_lossless}
\end{eqnarray}
In particular, when $m = \lbrace 0, \pi \rbrace$, $\phi_0$ fully recovers the EPR entanglement. However, if $m=\pm \frac{\pi}{2}$, $\phi$ has no effect on the system and EPR entanglement vanishes.
\end{lemma}
\begin{proof}
We have 
\small
\begin{eqnarray}
h(\phi) &=& \frac{4 b_2 \cos (m)(\cos(n + \phi)-\cos n )}{b_4-b_3\cos (2m)}, \label{eq:h_lossless} \\
h^{(1)}(\phi) &=& \frac{-4b_2\cos (m)\sin (n+\phi)}{b_4-b_3\cos (2m)}, \label{eq:h1_lossless} \\
h^{(2)}(\phi) &=& \frac{-4b_2\cos (m)\cos (n+\phi)}{b_4-b_3\cos (2m)}. \label{eq:h2_lossless}
\end{eqnarray}
\normalsize
The first derivative $ h^{(1)}(\phi)$ vanishes at  $\phi = \lbrace -n, \pm \pi -n \rbrace$.  As $b_2<0$ and $b_4-b_3\cos (2m) \geq b_4-b_3 >0$, we get
\begin{eqnarray}
h^{(2)}(-n) &=& \frac{-4b_2\cos (m)}{b_4-b_3\cos (2m)}\left\{ \begin{array}{ll}
 >0 & \textrm{if $m \in (-\frac{\pi}{2}, \frac{\pi}{2})$} \nonumber\\
 <0 & \textrm{if $m \in (-\pi, -\frac{\pi}{2})  \cup (\frac{\pi}{2}, \pi] $} ,
\end{array} \right.\\
h^{(2)}(\pm\pi-n) &=&\frac{4b_2\cos (m)}{b_4-b_3\cos (2m)} \left\{ \begin{array}{ll}
<0 & \textrm{if $m \in (-\frac{\pi}{2}, \frac{\pi}{2})$} \\
>0 & \textrm{if $m \in (-\pi, -\frac{\pi}{2})  \cup (\frac{\pi}{2}, \pi] $}. 
\end{array} \right.
\end{eqnarray}
Thus a local minimizer of $h(\phi)$ is
\begin{eqnarray}
\phi_0 = \left\{ \begin{array}{ll}
-n & \textrm{if $m \in (-\frac{\pi}{2}, \frac{\pi}{2})$} \\
\pm\pi-n & \textrm{if $m \in (-\pi, -\frac{\pi}{2})  \cup (\frac{\pi}{2}, \pi] $},  
\end{array} \right.
\end{eqnarray}
at which
\small
\begin{eqnarray}
h(-n) &=& \frac{4 b_2 \cos m (1-\cos n )}{b_4-b_3\cos (2m)}<0,  \textrm{if $m \in (-\frac{\pi}{2}, \frac{\pi}{2})$}, \nonumber\\
h(\pm \pi -n) &=& \frac{4 b_2 \cos m (-1-\cos n )}{b_4-b_3\cos (2m)}<0, \textrm{if $m \in (-\pi, -\frac{\pi}{2})  \cup (\frac{\pi}{2}, \pi] $}. 
\end{eqnarray}
\normalsize
We conclude that the above values of $\phi_0$ minimize the two-mode squeezing spectra influenced by phase shifts $\theta_1$ and $\theta_2$， when $m \neq \pm \frac{\pi}{2}$.
At $m= \pm \frac{\pi}{2}$, the term containing $\phi$ in $h(\phi)$ becomes $0$, thus $\phi$ has no impact on the two-mode squeezing spectra.

Define $V_\pm^{im}(m)$ as the two-mode squeezing spectra between the outputs of the system when $\phi=\phi_0$ with respect to variable $m$, according to (\ref{eq:V_lossless_mn}) and (\ref{eq: phi0_lossless}),
\small
\begin{eqnarray}
&&V_\pm^{im}(m) =\left\{ \begin{array}{ll}
2\frac{b_1+2b_2\cos \left(m \right)+b_3 \cos\left(2m \right)}{b_4-b_3\cos \left(2m\right)} & \textrm{if $m \in (-\frac{\pi}{2}, \frac{\pi}{2})$} \\
2\frac{b_1-2b_2\cos \left(m \right)+b_3 \cos\left(2m \right)}{b_4-b_3\cos \left(2m\right)} & \textrm{if $m \in (-\pi, -\frac{\pi}{2})  \cup (\frac{\pi}{2}, \pi] $}\\
2 & \textrm{if $m = \pm \frac{\pi}{2}$}. 
\end{array} \right. \label{eq:Vim_lossless}
\end{eqnarray}
\normalsize
Denote the first derivative of $V_\pm^{im}(m)$ as $V_\pm^{im(1)}(m)$. By applying Mathematica to solve $V_\pm^{im(1)}(m)=0$ based on (\ref{eq:b1-5}), we obtain that stationary points of $V_\pm^{im}(m)$ are $0$ and $\pi$. Values of $V_\pm^{im}(m)$ at the stationary points and the non-differentiable points $m= \pm \frac{\pi}{2}$ are
\small
\begin{eqnarray}
&&V_\pm^{im}(m)= \left\{ \begin{array}{ll}
2\frac{b_1+2b_2+b_3}{b_4-b_3}=V_\pm^{nops}, & \textrm{if $m= \lbrace 0, \pi \rbrace $} \\
2, & \textrm{if $m=\pm \frac{\pi}{2}$}, 
\end{array} \right.
\end{eqnarray}
\normalsize
hence, $V_\pm^{nops} \leq V_\pm^{im}(m) \leq 2$, which implies that at $m=\lbrace 0,  \pi\rbrace$, $\phi$ fully recovers the original EPR entanglement; at $m=\pm \frac{\pi}{2}$, $\phi$ has no effect on the EPR entanglement and the EPR entanglement vanishes; in remaining cases of $m$, $\phi$ improves the EPR entanglement impacted by $\theta_1$ and $\theta_2$ but cannot fully recover the EPR entanglement.
\end{proof}

Fig. \ref{fig: Vps_f_lossless} and Fig. \ref{fig: Vim_lossless} illustrate an example of the lossless dual-NOPA coherent feedback system undergoing phase shifts with $x=0.4$ and $y=1$, according to values reported in \cite{SN2014}. Note that in all the figures of two-mode squeezing spectra in the rest of the paper, values of squeezing spectra are given in dB unit, that is,  $V_\pm ({\rm dB})$  $=10\log_{10}(V_\pm)$. Hence, EPR entanglement exists when $V_\pm <10\log_{10}(2)=3.0103$ dB based on (\ref{eq:entanglement-criterion}) and the EPR entanglement is stronger as $V_\pm ({\rm dB})$ is more negative.

In Fig. \ref{fig: Vps_f_lossless}, the left plot shows that at some values of $m$ and $n$, $V_\pm^{ps}(m,n) >3.0103$ dB, which implies that phase shifts in the paths of the system can lead to death of EPR entanglement. The right plot shows the difference between values of $V_\pm^{ps}(m,n)$ and $V_\pm^{nops}$. When $(m,n) \neq (0,0)$, we see that $V_\pm^{ps}(m,n)>V_\pm^{nops}$, which indicates phase shifts in the paths between two NOPAs degrade the EPR entanglement.
\begin{figure}[htbp!]
\begin{center}
\includegraphics[scale=0.61]{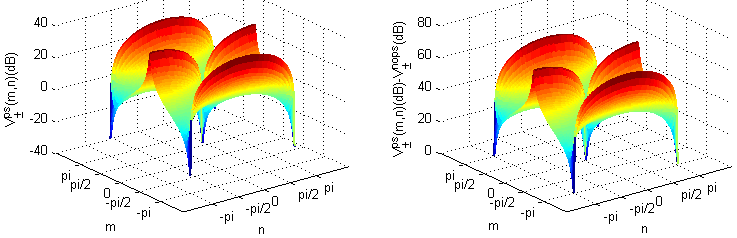}
\caption{Plots of $V_\pm^{ps}(m,n)({\rm dB})$ (left) and  $V_\pm^{ps}(m,n)({\rm dB})-V_\pm^{nops}({\rm dB})$ (right) of the lossless dual-NOPA coherent feedback system with $x=0.4$, $y=1$, $\alpha=1$, $\kappa=0$ and $\phi=0$.  }\label{fig: Vps_f_lossless}
\end{center}
\end{figure}
\begin{figure}[htbp!]
\begin{center}
\includegraphics[scale=0.67]{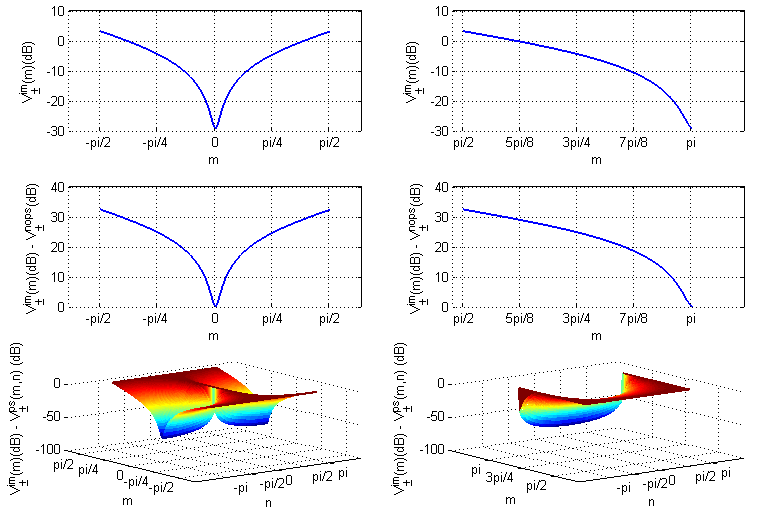}
\caption{Plots of $V_\pm^{im}(m)({\rm dB})$ (top row), $V_\pm^{im}(m)({\rm dB})-V_\pm^{nops}({\rm dB})$ (middle row) and $V_\pm^{im}(m)({\rm dB})-V_\pm^{ps}(m,n)({\rm dB})$ (bottom row)  of the lossless dual-NOPA coherent feedback system with $x=0.4$, $y=1$, $\alpha=1$ and $\kappa=0$. Ranges of values of $m$ are $[-\frac{\pi}{2}, \frac{\pi}{2}]$ (left column) and $[\frac{\pi}{2}, \pi]$ (right column).}\label{fig: Vim_lossless}
\end{center}
\end{figure}

Fig. \ref{fig: Vim_lossless} shows the effects of $\phi_0$ on the two-mode squeezing spectra. Note that as $V_\pm(0,m,n,\phi)$ is an even function of $m$, plots of $V_\pm(0,m,n,\phi)$ over intervals $(-\pi, -\frac{\pi}{2}]$ and $[\frac{\pi}{2}, \pi]$ of $m$ are symmetric, thus we do not show the plots of squeezing spectra versus varying values of $m$ ranging from $-\pi$ to $-\frac{\pi}{2}$.
The plots of $V_\pm^{im}(m)$ against the parameter $m$ in the top row shows that EPR entanglement exists ($V_\pm^{im}(m)<3.0103$ dB) over the range of $m$, except for $m=\pm \frac{\pi}{2}$ ($V_\pm^{im}(\pm \frac{\pi}{2})=3.0103$ dB). 
The middle row illustrates the original EPR entanglement is fully recovered by $\phi_0$ at $m = \lbrace 0, \pi \rbrace$.
The bottom row displays the difference between values of $V_\pm^{im}(m)$ and $V_\pm^{ps}(m,n)$ against the parameters $m$ and $n$. We see that the difference value is not positive which implies that, $\phi_0$ improves the two-mode squeezing spectra in most scenarios, but does not impact the system when $(m=\pm \frac{\pi}{2}, n\in(-\pi, \pi])$, $(m \in (-\frac{\pi}{2}, \frac{\pi}{2}), n=0)$ and $(m\in (-\pi, -\frac{\pi}{2})\cup(\frac{\pi}{2},\pi], n=\pi)$. Note that based on Lemma \ref{le: phi_lossless}, in the first case where $(m=\pm \frac{\pi}{2}, n\in(-\pi, \pi])$, the adjustable phase shifters at the outputs do not impact the EPR entanglement of the system for any values of $\phi$; however for $(m\in (-\pi, -\frac{\pi}{2})\cup(\frac{\pi}{2},\pi], n=\pi)$, though $\phi_0$ does not have an effect on the EPR entanglement impacted by $\theta_1$ and $\theta_2$, $\phi=\phi_0$ is the best choice based on the proof of Lemma \ref{le: phi_lossless}.

\subsection{Effect of phase shifts on EPR entanglement of the dual-NOPA coherent feedback system with losses}
\label{sec:entanglement-loss}
Now let us investigate the performance of the dual-NOPA coherent feedback system under the presence of phase shifts, transmission losses and amplification losses. 

Let $m=\frac{\theta_1-\theta_2}{2}$, $n=\frac{\theta_1+\theta_2}{2}$, $m, n, m+n, n-m \in (-\pi, \pi]$ and $\phi=\phi_1+\phi_2$, $\phi\in (-2\pi, 2\pi]$,  the two-mode squeezing spectra between the two outputs in the dual-NOPA coherent feedback system under the effect of phase shifts and losses is
\small
\begin{eqnarray}
&& V_\pm(0, m, n, \phi) =2\frac{c_1+2c_2\cos \left(m \right)\cos \left( n+\phi \right)+c_3 \cos\left(2m \right)}{c_4-c_5\cos \left(2m \right)}, \label{eq: V_loss}
\end{eqnarray} 
\normalsize where
\small
\begin{eqnarray}
c_1 &=& 4 \alpha^4 \epsilon^2 \gamma^2 (\kappa^4 + 8 \kappa \epsilon^2 \gamma + (\epsilon^2 - \gamma^2)^2 - 2 \kappa^2 (\epsilon^2 + \gamma^2)) + \alpha^2 (-\kappa^2 + \epsilon^2 + \gamma^2)^2 (\kappa^4   \nonumber \\
&& + 4 \kappa^3 \gamma - 2 \kappa^2 (\epsilon^2 - 3 \gamma^2) + 4 \kappa \gamma (\epsilon^2 + \gamma^2) + \epsilon^4 + 2 (1 + 2 \beta^2) \epsilon^2 \gamma^2 + \gamma^4)  \nonumber \\
&&  + (\kappa^2 - \epsilon^2 + 2 \kappa \gamma + \gamma^2)^2 (4 \gamma (\kappa + \gamma) (\kappa^2 + \epsilon^2 + \kappa \gamma)+ \beta^2 (-\kappa^2 + \epsilon^2 + \gamma^2)^2), \nonumber \\
c_2 &=& 4 \alpha \epsilon \gamma (-\kappa^2 + \epsilon^2 + \gamma^2) (-\kappa^4 - 4 \kappa^3 \gamma - 6 \kappa^2 \gamma^2  + 4 \alpha^2 \kappa \epsilon^2 \gamma + \epsilon^4  - 4 \kappa \gamma^3 - \gamma^4),\nonumber \\
c_3 &=& 8 \alpha^2 \epsilon^2 \gamma^2 (-\kappa + \epsilon - \gamma)  (\kappa + \epsilon + \gamma) (3 \kappa^2  + 2 \kappa \gamma + \epsilon^2 - \gamma^2),\nonumber \\
c_4 &=& \kappa^8  + 8 \kappa^7 \gamma - 4 \kappa^6 (\epsilon^2 - 7 \gamma^2) - 8 \kappa^5 (3 \epsilon^2 \gamma - 7 \gamma^3)+ \kappa^4 (6 \epsilon^4 - 60 \epsilon^2 \gamma^2 + 70 \gamma^4)  \nonumber \\
&&  + 8 \kappa^3 (3 \epsilon^4 \gamma - 10 \epsilon^2 \gamma^3 + 7 \gamma^5) - 4 \kappa^2 (\epsilon^2 - 7 \gamma^2) (\epsilon^2 - \gamma^2)^2 + 8 \kappa \gamma (-\epsilon^2 + \gamma^2)^3 \nonumber \\
&&  + \epsilon^8 - 4 \epsilon^6 \gamma^2 + 2 (3 + 8 \alpha^4) \epsilon^4 \gamma^4- 4 \epsilon^2 \gamma^6 + \gamma^8,\nonumber \\
c_5 &=& 8 \alpha^2 \epsilon^2 \gamma^2 (\kappa^2 + 2 \kappa \gamma - \epsilon^2  + \gamma^2)^2.\label{eq: c1-c5}
\end{eqnarray}
\normalsize

Similar to Section \ref{sec:general_case_lossless}, we have the following lemmas.
\begin{lemma}
\label{le: f_loss} 
The presence of the phase shifts $\theta_1\neq0$ and $\theta_2\neq0$ degrades the two-mode squeezing spectra ($V_\pm^{ps}(m,n)>V_\pm^{nops}$), thus degree of EPR entanglement becomes worse or EPR entanglement may vanish.
\end{lemma}
\begin{proof}
Based on the functions defined at the beginning of Section \ref{sec:entanglement-loss}, we have 
\small
\begin{eqnarray}
f(m,n) &=& 2\frac{c_1+ 2c_2 \cos (m) \cos(n)+ c_3 \cos(2m)}{c_4-c_5\cos (2m)} -2\frac{c_1+ 2c_2+ c_3 }{c_4-c_5}, \label{eq:f_loss} \\
\frac{\partial f}{\partial n}(m,n) &=& \frac{-4 c_2 \cos(m) \sin(n)}{c_4 - c_5 \cos(2 m)}, \label{eq:f1_loss} \\
\frac{\partial f}{\partial m}(m,n) &=& -4 \sin(m) \frac{c_2 \cos(n) (c_4 + c_5 \sin(m)^2)}{(c_4 -c_5 \cos(2 m))^2} -4 \sin(m) \frac{ 3 c_2 c_5 \cos(m)^2 \cos(n) }{(c_4 - c_5 \cos(2 m))^2} \nonumber \\
&& \quad -4 \sin(m) \frac{2 (c_3 c_4 + c_1 c_5) \cos(m)}{(c_4 - c_5 \cos(2 m))^2}. \label{eq:f2_loss}  
\end{eqnarray}
\normalsize
Similar to the proof in Section \ref{sec:general_case_lossless}, global minima of $f(m,n)$ are stationary points. As given by Mathematica, the first order partial derivatives of $f$ with respect to the variable $m$ and $n$ vanish at $(m,n)=\lbrace(0, 0), (0, \pi), (\pm \frac{\pi}{2}, \frac{\pi}{2}), (- \frac{\pi}{2}, -\frac{\pi}{2}) \rbrace$, at which values of $f(m.n)$ are
\small
\begin{eqnarray}
&&f(m,n) = \left\{ \begin{array}{l}
0, \quad \textrm{if $(m,n)=(0, 0)$, that is, $\theta_1=\theta_2=0$} \\
\frac{-8c_2}{c_4-c_5}, \quad \textrm{if $(m,n)=(0, \pi)$} \\
\frac{-4(c_1c_5+c_3c_4+c_2c_4+c_2c_5)}{(c_4-c_5)(c_4+c_5)}, \quad \textrm{if $(m,n)=\lbrace(\pm \frac{\pi}{2}, \frac{\pi}{2}), (- \frac{\pi}{2}, -\frac{\pi}{2}) \rbrace$}.
\end{array} \right.
\end{eqnarray}
\normalsize
Based on stability condition (\ref{eq:stability-cond}), replacing $\epsilon$, $\gamma$ and $\kappa$ in (\ref{eq: c1-c5}) with definitions $\epsilon=x \gamma_r$, $\gamma=\frac{\gamma_r}{y}$, $\kappa=\frac{3 \times 10^6}{\sqrt{2} \times 0.6}x$, $0<x,y,\alpha\leq1$, $\gamma_r=7.2 \times 10^7$ and noting $\gamma\geq \epsilon\geq \alpha\epsilon$, Mathematica gives that 
\small
\begin{eqnarray}
&&c_4-c_5=\left(\kappa^4 + \epsilon^4 + 4 \kappa^3 \gamma - 2 (1 + 2 \alpha^2) \epsilon^2 \gamma^2 + \gamma^4 - 2 \kappa^2  (\epsilon^2 - 3 \gamma^2)+ \kappa (-4 \epsilon^2 \gamma + 4 \gamma^3)\right)^2,\nonumber \\
&&c_2 = 4\alpha\epsilon\gamma d_1 d_2,\nonumber\\
&&c_1c_5+c_3c_4+c_2c_4+c_2c_5=4 \alpha \epsilon \gamma (-\kappa^2 + \epsilon^2 - 2 \kappa \gamma + 2 \alpha \epsilon \gamma - \gamma^2)^2d_1d_2d_3,
\end{eqnarray}
where
\begin{eqnarray}
&&d_1=\epsilon^2 +\gamma^2-\kappa^2=\left(5.1715\times 10^{15} x^2 + \frac{5.184\times 10^{15}}{y^2}\right)>0, \nonumber\\
&&d_2=-(\gamma^4-\epsilon^4 )-4 \kappa \gamma(\gamma^2-\alpha^2 \epsilon^2) - 6 \kappa^2 \gamma^2  - 4 \kappa^3 \gamma -\kappa^4  <0, \nonumber \\
&&d_3=(\kappa^2+\gamma^2-\epsilon^2)^2+4\kappa^2\gamma^2+4\alpha^2\epsilon^2\gamma^2+4\kappa^3\gamma+4\kappa\gamma(\gamma^2-\epsilon^2)>0. \label{eq:d1-d3}
\end{eqnarray}
\normalsize
We see that $c_4-c_5>0$, $c_2<0$ and $c_1c_5+c_3c_4+c_2c_4+c_2c_5<0$. Therefore, $f(m,n)\geq 0$, that is, $V_\pm^{ps}(m,n) \geq V_\pm^{nops}$. Equality holds when $(m,n)=(0, 0)$, which is the case with no phase shifts. We obtain Lemma \ref{le: f_loss}.
\end{proof}
\begin{lemma}
\label{le: phi_loss} 
$\phi$ minimizes the two-mode squeezing spectra at $\omega=0$ impacted by $\theta_1$ and $\theta_2$ if its value is set as
\begin{eqnarray}
\phi_0 = \left\{ \begin{array}{ll}
-n & \textrm{if $m \in (-\frac{\pi}{2}, \frac{\pi}{2})$}\\
\pm \pi - n & \textrm{if $m \in (-\pi, -\frac{\pi}{2})  \cup (\frac{\pi}{2}, \pi] $},
\end{array} \right.
\end{eqnarray} 
However, when $m=\pm \frac{\pi}{2}$, $\phi$ has no effect on the system.
\end{lemma}
\begin{proof}
We have 
\small
\begin{eqnarray}
h(\phi) &=& \frac{4 c_2 \cos m (\cos(n + \phi)-\cos n )}{c_4-c_5\cos (2m)}, \label{eq:h_loss} \\
h^{(1)}(\phi) &=& \frac{-4c_2\cos (m)\sin (n+\phi)}{c_4-c_5\cos (2m)}, \label{eq:h1_loss} \\
h^{(2)}(\phi) &=& \frac{-4c_2\cos (m)\cos (n+\phi)}{c_4-c_5\cos (2m)}. \label{eq:h2_loss} 
\end{eqnarray}
\normalsize
As proof of Lemma \ref{le: phi_lossless}, we obtain a local minimizer $\phi_0$.
\end{proof}
\begin{lemma}
\label{le: Vim_loss} 
Define $V_\pm^{im}(m)$, the two-mode squeezing spectra at $\omega=0$ of the outgoing fields in the dual-NOPA coherent feedback system with $\phi=\phi_0$ as a function of $m\,$\footnote[2]{Eventhough there is more than one minimum $\phi_0$, the function $V_\pm^{im}(m)$ takes the same value for all minima.},
\small
\begin{eqnarray}
&&V_\pm^{im}(m)= \left\{ \begin{array}{ll}
2\frac{c_1+2c_2\cos \left(m \right)+c_3 \cos\left(2m \right)}{c_4-c_5\cos \left(2m\right)} & \textrm{if $m \in (-\frac{\pi}{2}, \frac{\pi}{2})$} \\
2\frac{c_1-2c_2\cos \left(m \right)+c_3 \cos\left(2m \right)}{c_4-c_5\cos \left(2m\right)} & \textrm{if $m \in (-\pi, -\frac{\pi}{2})  \cup (\frac{\pi}{2}, \pi] $} \\
2\frac{c_1-c_3}{c_4+c_5} & \textrm{if $m = \pm \frac{\pi}{2}$}.
\end{array}  \right. \label{eq: Vim_loss}
\end{eqnarray}
\normalsize
$V_\pm^{im}(m)=2$ has four real roots denoted by $\pm m_1$ and $\pm m_2$, with $m_2 \geq m_1 \geq 0$.
EPR entanglement under the influence of phase shifts and losses exists on intervals $(-\pi, -m_2)$,  $(-m_1, m_1)$ and $(m_2, \pi]$. The original EPR entanglement impacted by $\theta_1$ and $\theta_2$ is fully recovered by $\phi_0$ if $m= \lbrace 0, \pi \rbrace$. Also, EPR entanglement is improved as value of $m$ approaches $\lbrace 0, \pm\pi \rbrace$.
\end{lemma}
\begin{proof}
The first derivative of $V_\pm^{im}(m)$ is
\small
\begin{eqnarray}
V_\pm^{im(1)}(m)= \left\{ \begin{array}{ll}
-4\frac{\left(2 \left(c_3 c_4 + c_1 c_5\right) \cos (m) + c_2 \left(c_4 + 2 c_5 + c_5 \cos(2 m)\right)\right) \sin(m)}{\left(c_4-c_5\cos (2m)\right)^2}, & \textrm{if $m \in (-\frac{\pi}{2}, \frac{\pi}{2})$} \\
4\frac{\left(-2 \left(c_3 c_4 + c_1 c_5\right) \cos (m) + c_2 \left(c_4 + 2 c_5 + c_5 \cos(2 m)\right)\right) \sin(m)}{\left(c_4-c_5\cos (2m)\right)^2}, & \textrm{if $m \in (-\pi, -\frac{\pi}{2})  \cup (\frac{\pi}{2}, \pi] $}. 
\end{array} \right.
\end{eqnarray}
\normalsize
Employing (\ref{eq: c1-c5}) and solving $V_\pm^{im(1)}(m)=0$ via Mathematica, we obtain that stationary points of $V_\pm^{im}(m)$ are $0$ and $\pi$. Values of $V_\pm^{im}(m)$ at stationary points and non-differentiable points $\pm \frac{m}{2}$ are
\small
\begin{eqnarray}
&&V_\pm^{im}(m)= \left\{ \begin{array}{ll}
2\frac{c_1+2c_2+c_3}{c_4-c_5}=V_\pm^{nops}, & \textrm{if $m= \lbrace 0, \pi \rbrace $} \\
2\frac{c_1-c_3}{c_4+c_5}, & \textrm{if $m=\pm \frac{\pi}{2}$}.
\end{array} \right.
\end{eqnarray}
\normalsize
Noting $0<x, y, \alpha\leq1$ and $d_3$ in (\ref{eq:d1-d3}). Mathematica shows that
\small
\begin{eqnarray}
&&(c_1-c_3)-(c_4+c_5)=8\epsilon^2 \gamma \left((1+ \alpha^2)\kappa+(1- \alpha^2)\gamma\right) d_3>0.
\end{eqnarray}
\normalsize
Hence, $2\frac{c_1-c_3}{c_4+c_5}>2$. Consequently, global minima of $V_\pm^{im}(m)$ are at $m= \lbrace 0, \pi \rbrace $ at which the original EPR entanglement is fully recovered. 

Recall $d_1$ and $d_2$ from (\ref{eq:d1-d3}). Mathematica then gives
\small
\begin{eqnarray}
&&V_\pm^{im(1)}(m)= \left\{ \begin{array}{ll}
\frac{-d_4\sin(m)}{d_5}, & \textrm{if $m \in (-\frac{\pi}{2}, \frac{\pi}{2})$} \\
\frac{d_4\sin(m)}{d_6}, & \textrm{if $m \in (-\pi, -\frac{\pi}{2})  \cup (\frac{\pi}{2}, \pi] $}. 
\end{array} \right.
\end{eqnarray}
\normalsize
where
\small
\begin{eqnarray}
&&d_4=16\alpha\epsilon\gamma d_1 d_2 <0, \nonumber\\
&&d_5=\left(\kappa^4 + \epsilon^4 + 4 \kappa^3 \gamma - 2 \epsilon^2 \gamma^2 + 4 \alpha^2 \epsilon^2 \gamma^2 + \gamma^4 - 2 \kappa^2 (\epsilon^2 - 3 \gamma^2) + \kappa (-4 \epsilon^2 \gamma + 4 \gamma^3)\right.\nonumber \\
&&\quad \left. + 4 \alpha \epsilon \gamma (\kappa^2 - \epsilon^2 + 2 \kappa \gamma + \gamma^2) \cos(m)\right)^2>0, \nonumber\\
&&d_6=\left(\kappa^4 + \epsilon^4 + 4 \kappa^3 \gamma - 2 \epsilon^2 \gamma^2 + 4 \alpha^2 \epsilon^2 \gamma^2 + \gamma^4 - 2 \kappa^2 (\epsilon^2 - 3 \gamma^2) + \kappa (-4 \epsilon^2 \gamma + 4 \gamma^3)\right.\nonumber \\
&&\quad \left. - 4 \alpha \epsilon \gamma (\kappa^2 - \epsilon^2 + 2 \kappa \gamma + \gamma^2) \cos(m)\right)^2>0.
\end{eqnarray}
\normalsize
Therefore,
\small
\begin{eqnarray}
&&V_\pm^{im(1)}(m)= \left\{ \begin{array}{ll}
>0, & \textrm{if $m \in (-\pi, -\frac{\pi}{2})\cup (0, \frac{\pi}{2})$} \\
<0, & \textrm{if $m \in (-\frac{\pi}{2},0) \cup (\frac{\pi}{2}, \pi)$}, 
\end{array} \right.
\end{eqnarray}
\normalsize
implies that $V_\pm^{im}(m)$ is a piecewise monotonically increasing function on intervals $(-\pi, -\frac{\pi}{2})\cup (0, \frac{\pi}{2})$ and a piecewise monotonically decreasing function over $ (-\frac{\pi}{2},0) \cup (\frac{\pi}{2}, \pi)$. It approaches the maximum value $2\frac{c_1-c_3}{c_4+c_5}>2$ at $m=\pm \frac{\pi}{2}$ and minimum value equals to $V_\pm^{nops}$ when $m= \lbrace 0, \pi \rbrace $. Hence, the even function $V_\pm^{im}(m)$ of $m$ has four real roots denoted by $\pm m_1$ and $\pm m_2$ with $0\leq m_1 \leq m_2$. $V_\pm^{im}(m)<2$ on intervals  $(-\pi, -m_2)$,  $(m_1, m_1)$ and $(m_2, \pi]$. Proof is completed.
\end{proof}

\begin{figure}[htbp!]
\begin{center}
\includegraphics[scale=0.55]{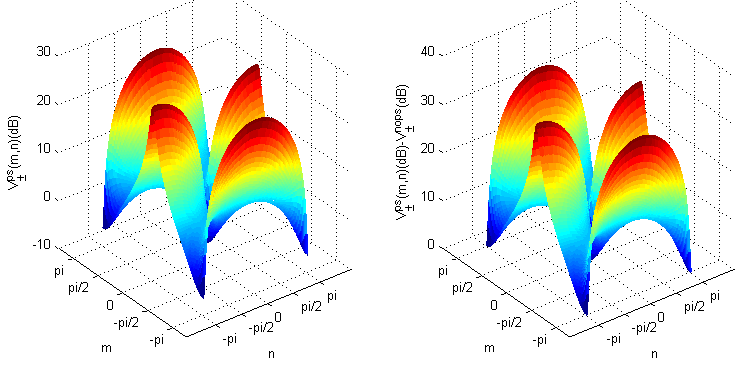}
\caption{Plots of $V_\pm^{ps}(m,n)({\rm dB})$ (left) and  $V_\pm^{ps}(m,n)({\rm dB})-V_\pm^{nops}({\rm dB})$ (right) of the dual-NOPA coherent feedback system with $x=0.4$, $y=1$, $\alpha=0.95$, $\kappa=\frac{3 \times 10^6}{\sqrt{2} \times 0.6}x$ and $\phi=0$. }\label{fig: Vps_f_loss}
\end{center}
\end{figure}
\begin{figure}[htbp!]
\begin{center}
\includegraphics[scale=0.52]{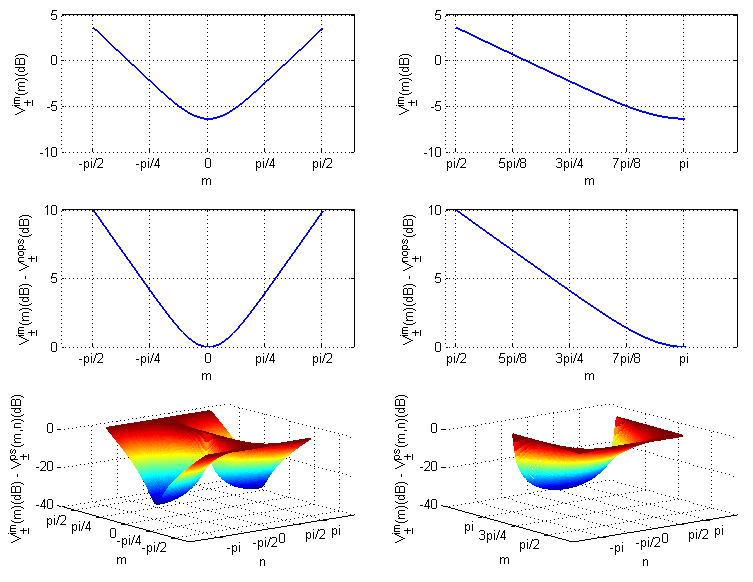}
\caption{Plots of $V_\pm^{im}(m)({\rm dB})$ (top row), $V_\pm^{im}(m)({\rm dB})-V_\pm^{nops}({\rm dB})$ (middle row) and $V_\pm^{im}(m)({\rm dB})-V_\pm^{ps}(m,n)({\rm dB})$ (bottom row) of the dual-NOPA coherent feedback system with $x=0.4$, $y=1$, $\alpha=0.95$ and $\kappa=\frac{3 \times 10^6}{\sqrt{2} \times 0.6}x$.  Ranges of values of $m$ are $[-\frac{\pi}{2}, \frac{\pi}{2}]$ (left column) and $[\frac{\pi}{2}, \pi]$ (right column). }\label{fig: Vim_loss}
\end{center}
\end{figure}

Fig. \ref{fig: Vps_f_loss} and Fig. \ref{fig: Vim_loss} illustrate an example of the dual-NOPA coherent feedback system undergoing both phase shifts and losses with $x=0.4$, $y=1$ and $\alpha=0.95$. Similar to Fig. \ref{fig: Vps_f_lossless}, the left plot in Fig. \ref{fig: Vps_f_loss} shows that EPR entanglement vanishes at some values of $m$ and $n$. The right plot shows that the non-zero phase shifts in the paths decrease the degree of EPR entanglement.

Fig. \ref{fig: Vim_loss} illustrates the effect of $\phi_0$. Based on symmetric property of function $V_\pm(0, m, n, \phi)$, we can see from Fig. \ref{fig: Vim_loss} that the top row shows that under the effect of $\phi_0$, for some values of $m$ near $\pm \frac{\pi}{2}$ there is no EPR entanglement between the two outgoing fields ($V_\pm^{im}(m) \geq 3.0103$ dB);
the middle row shows the original EPR entanglement is fully recovered at $m = \lbrace 0, \pi \rbrace$ and the bottom row shows that $\phi_0$ improves the two-mode squeezing spectra except for the cases where ($m=\pm \frac{\pi}{2}, n\in(-\pi, \pi])$, $(m\in(-\frac{\pi}{2}, \frac{\pi}{2}), n=0)$ and $(m\in (-\pi, -\frac{\pi}{2})\cup(\frac{\pi}{2}, \pi], n=\pi)$. Note that based on Lemma \ref{le: phi_loss}, any value of $\phi$ does not impact the EPR entanglement of the system when $m=\pm \frac{\pi}{2}$; while $\phi=\phi_0$ is the best option in the last two scenarios where  $(m\in(-\frac{\pi}{2}, \frac{\pi}{2}), n=0)$ and $(m\in (-\pi, -\frac{\pi}{2})\cup(\frac{\pi}{2}, \pi], n=\pi)$.

Table \ref{tb: transmission} and Table \ref{tb: amplification} illustrate the effect of transmission and amplification losses on the existence of EPR entanglement with an optimal choice of $\phi_0$. We see that as either transmission losses or amplification losses increase, the range of values of $m$ over which the EPR entanglement does not exist becomes larger, and the performance of EPR entanglement worsens in the presence of losses, as can be expected.
\begin{table}[htbp]
\centering
\caption{Influence of transmission losses on the range of nonexistence of EPR entanglement with $x=0.4$, $y=1$ and $\kappa=\frac{3 \times 10^6}{\sqrt{2} \times 0.6}x$}\label{tb: transmission}
\begin{tabular}{|c|c|c|}
\hline
$\alpha$ & $[-m_2, -m_1]$ & $[m_1, m_2]$ \\
\hline
$1$ &  $[-1.58951, -1.55208]$ &  $[1.55208, 1.58951]$ \\
\hline
$0.97$ &  $[-1.61856,  -1.52303]$   &  $[1.52303, 1.61856]$     \\
\hline
$0.95$ &  $[-1.63848, -1.50311]$    &    $[1.50311,  1.63848]$    \\
\hline
\end{tabular}
\end{table}
\begin{table}[htbp]
\centering
\caption{Influence of amplification losses on the range of nonexistence of EPR entanglement with $x=0.4$, $y=1$ and $\alpha=0.95$}\label{tb: amplification}
\begin{tabular}{|c|c|c|}
\hline
$\kappa$ & $[-m_2, -m_1]$ & $[m_1, m_2]$ \\
\hline
$0.1\frac{3 \times 10^6}{\sqrt{2} \times 0.6}x$ &  $[-1.62156, -1.52003]$ &  $[1.52003, 1.62156]$ \\
\hline
$0.2\frac{3 \times 10^6}{\sqrt{2} \times 0.6}x$ &  $[-1.62344,  -1.51815]$   &  $[1.51815, 1.62344]$     \\
\hline
$0.5\frac{3 \times 10^6}{\sqrt{2} \times 0.6}x$ &  $[-1.62907,  -1.51252]$   &  $[1.51252, 1.62907]$   \\
\hline
$\frac{3 \times 10^6}{\sqrt{2} \times 0.6}x$ &  $[-1.63848, -1.50311]$    &    $[1.50311,  1.63848]$    \\
\hline
\end{tabular}
\end{table}

\section{Conclusion}
\label{sec:conclusion}
This paper has investigated the effects of phase shifts on stability and EPR entanglement of a dual-NOPA coherent feedback network. Stability condition determined by parameters of the system with losses and phase shifts is derived. The system remains stable in the presence of phase shifts, whenever the system is stable in the absence of phase shifts.

In the lossless system, in the absence of transmission and amplification losses, the presence of phase shifts $\theta_1\neq0$ and $\theta_2\neq0$ in the paths between two NOPAs degrades the two-mode squeezing spectra between the two outputs in the system, which implies EPR entanglement worsens or even vanishes. The two-mode squeezing spectra under the influence of $\theta_1$ and $\theta_2$ is minimized by setting $\phi=\phi_0$. However, existence of EPR entanglement and the degree of EPR entanglement recovered by $\phi_0$ depend on the parameter $m$.  EPR entanglement is fully recovered by $\phi_0$ if $m= \lbrace 0, \pi \rbrace$. EPR entanglement vanishes when $m=\pm\frac{m}{2}$.

When transmission and amplification losses are not neglected, the two-mode squeezing spectra are degraded by phase shifts in the paths and are maximally recovered by setting $\phi=\phi_0$. However, existence of EPR entanglement is impacted by both phase shifts and losses in the paths. The range of values of $m$ over which the EPR entanglement can be improved by $\phi_0$ decreases as losses grow.


\begin{thebibliography}{}

\bibitem{GardinerBook}
C. W. Gardiner and P. Zoller, 
{\it Quantum Noise}, 
(Springer-Verlag, Berlin and New York, 3rd edition, 2004).

\bibitem{BL2005}
S. L. Braunstein and P. van Loock,
Quantum information with continuous variables,
Rev. Mod. Phys. 77, 513-577 (2005).

\bibitem{BSLR}
W. P. Bowen, R. Schnabel, P. K. Lam and T. C. Ralph,
A characterization of continuous variable entanglement,
Phys. Rev. A 69, 012304 (2004).

\bibitem{Ou1992}
Z. Y. Ou, S. F . Pereira, and H. J. Kimble, 
Realization of the Einstein-Podolski-Rosen paradox for continuous 
variables in nondegenerate parametric amplification, 
Appl. Phys. B 55, 265 (1992).


\bibitem{SN2014}
Z. Shi and H. I. Nurdin,
Coherent feedback enabled distributed generation of entanglement between propagating Gaussian fields,
to appear in Quantum Information Processing (2014). [Online] Available: http://dx.doi.org/10.1007/s11128-014-0845-4.

\bibitem{Laurat2005}
J. Laurat, G. Keller, J.A. Oliveira-Huguenin, C. Fabre, T. Coudreau, A. Serafini, G. Adesso and F. Illuminati, 
Entanglement of two-mode Gaussian states: characterization and experimental production and manipulation,
J. Opt. B: Quantum Semiclass. Opt. 7, S577-S587 (2005)

\bibitem{Belavkin2008}
V. P. Belavkin and S. C. Edwards, 
Quantum filtering and optimal control, 
in Quantum Stochastics and Information: 
Statistics, Filtering and Control, 143-205, 
(World Scientific, 2008). 

\bibitem{WisemanBook}
H. M. Wiseman and G. J. Milburn, 
{\it Quantum Measurement and Control}, 
(Cambridge University Press, 2010). 

\bibitem{Vitali2006}
D. Vitali, G. Morigi, and J. Eschner, 
Single cold atom as efficient stationary source of 
EPR entangled light, 
PRA 74, 053814 (2006).

\bibitem{NY2012}
H. I. Nurdin and N. Yamamoto,
Distributed entanglement generation between continuous-mode Gaussian fields with measurement-feedback enhancement,
Phys. Rev. A 86, 022337 (2012).

\bibitem{GJN2010}
J. E. Gough, M. R. James and H. I. Nurdin,
Squeezing components in linear quantum feedback networks,
Phys. Rev. A 81, 023804 (2010).

\bibitem{NJD2009}
H. I. Nurdin, M. R. James and A. C. Doherty,
Network Synthesis of Linear Dynamical Quantum Stochastic Systems,
SIAM J. Control Optim., 48(4), 2686–2718 (2009).

\bibitem{GK2005}
C. C. Gerry and P. L. Knight,
Introductory Quantum Optics,
(Cambridge University Press, 2005).

\bibitem{Iida2012}
S. Iida, M. Yukawa, H. Yonezawa, N. Yamamoto, and A. Furusawa, 
Experimental demonstration of coherent feedback control 
on optical field squeezing, 
IEEE Trans. Automat. Contr. 57(8), 2045-2050 (2012). 

\end{thebibliography}
\end{document}